\documentclass[11pt]{article}
\pdfoutput=1

\newif\ifFull
\Fulltrue
\usepackage{amsmath,amssymb}
\usepackage{graphicx,epstopdf}
\usepackage{url,hyperref}
\usepackage{jeffe}
\usepackage{color}
\usepackage{microtype}
\usepackage {plaatjes}
\usepackage{fullpage}

\usepackage {xspace}

\definecolor {infocolor} {rgb} {0.6,0.6,0.6}
\definecolor {sepia} {rgb} {0.75,0.30,0.15}
\definecolor {remarkcolor} {rgb} {0.60,0.30,0.30}
\newcommand {\fsg} [1] {\ensuremath {G \langle #1 \rangle}}

\newtheorem {observation} {Observation}
\newtheorem {lemma} {Lemma}
\newtheorem {theorem} {Theorem}
\newtheorem {corollary} {Corollary}

\ifFull\else
\renewcommand{\subsection}[1]{\subsubsection{#1.}}
\fi

\newcommand {\plaatjesh} [3]
{
  \begin {figure} [h]
    \centering
    #1
    \hfill
    \mbox {}
    \textsf{\caption {\small #2}
    \label {fig:#3}}
  \end {figure}
}
\newcommand {\drieplaatjesh} [5] []
{
  \plaatjesh
  {
    \subplaatje {#1} {#2}
    \subplaatje {#1} {#3}
    \subplaatje {#1} {#4}
  } {#5} {#2+#3+#4}
}

\title{Drawing Graphs in the Plane \\ with
a Prescribed Outer Face and Polynomial Area}

\author{
  Erin W. Chambers\footnote
  { Dept. of Math and Computer Science, Saint Louis Univ., USA.
    \texttt{echambe5(at)slu.edu}
  }
  \and
  David Eppstein\footnote
  {
    Computer Science Dept., University of California, Irvine, USA.
    \texttt{\{eppstein,goodrich,loffler\}(at)ics.uci.edu}
  }
  \and
  Michael T. Goodrich\footnotemark [2]
  \and
  Maarten L\"offler\footnotemark [2]
}
\date{ }

\begin{document}

\maketitle

\begin{abstract}
We study the classic graph drawing problem of drawing a planar graph
using straight-line edges with a prescribed convex polygon
as the outer face. Unlike previous algorithms for this
problem, which may produce drawings with exponential area,
our method produces drawings with polynomial area.
In addition, we allow for collinear points on the boundary, provided
such vertices do not create overlapping edges.
Thus, we solve an open problem of Duncan {\it et al.}, which,
when combined with their work, implies that we can produce
a planar straight-line drawing of a combinatorially-embedded genus-$g$ graph
with the graph's canonical polygonal schema drawn as a convex polygonal
external face.
\end{abstract}

%%%%%%%%%%%%%%%%%%%%%%%%%%%%%%%%%%%%%%%%%%%%%%%%
\section{Introduction}
\label{sec:intro}
The study of planar graphs has been a driving force for graph theory,
graph algorithms, and graph drawing.
Our interest, in this paper, is on methods for drawing planar graphs
without edge crossings
using straight line segments for edges, in such a way that all faces are convex
polygons and the outer face is a given shape.
Figure~\ref{fig:example} shows an example.

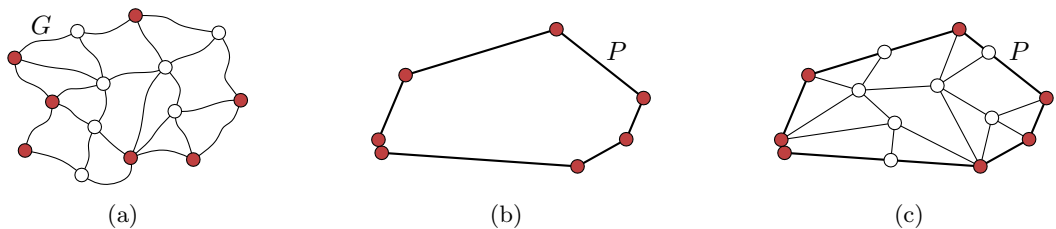
\drieplaatjesh {intro-input-G} {intro-input-P} {intro-output} {Our problem: given (a) a combinatorially embedded planar graph $G$ and (b)~a polygon $P$ with certain vertices on the outer face of $G$ marked as corresponding to vertices of $P$, find (c)~a straight-line embedding of $G$ that uses $P$ as the shape of its outer face.\label{fig:example}}

\subsection{Related Prior Work}
In seminal work that has been highly influential in the graph drawing
literature, Tutte~\cite{t-crg-60,t-hdg-63}
shows that it is possible to draw any planar graph, $G$, using
non-crossing straight-line edges so that the vertices of the
outer face are drawn on the boundary of a prescribed convex polygon.
This work has influenced a host of subsequent papers, and,
according to Google Scholar, Tutte's 1963 paper
has been directly cited over 600 times.
For instance, along with seminal work on drawing maximal planar
graphs~\cite{f-slrpg-48,s-cm-51,w-bzv-36} (with triangular outer faces),
his work has influenced many researchers
to study methods for
drawing planar graphs using straight-line edges (e.g.,
see~\cite{cgt-cdgtt-96,ck-cgd3c-97,cp-ltadp-95,%
dpp,d-gdt-10,k-dpguc-96,br-scdpg-06,s-epgg-90}).
Moreover, not only has Tutte's result itself been highly influential,
but because his method is based on a force-directed layout method,
his work has also influenced a considerable amount of work on
force-directed layouts
(e.g., see~\cite{dh-dgnus-96,dett-gd-99,fr-gdfdp-91,ggk-mdafd-04,sm-gdmsm-95}).

Unfortunately,
one of the drawbacks of Tutte's algorithm is that it
can result in drawings with exponential area.
This area blowup is not an inherent requirement for planar
straight-line drawings, however, as
a number of researchers have shown that it is possible to
produce drawings of planar graphs with non-crossing straight-line
edges using polynomial area
(e.g.,~\cite{ck-cgd3c-97,cp-ltadp-95,dpp,k-dpguc-96,s-epgg-90}).
Nevertheless, all of these straight-line drawing algorithms lose a
critical feature of Tutte's drawing algorithm, in that none of them
allow for the vertices of a planar graph's outer
face to be placed on the boundary of a prescribed
convex polygon.
Becker and Hotz~\cite{bh-olpgf-87}, on the other hand,
show how to draw a planar graph with
a prescribed outer face so as to optimize the total weighted edge
length, but, like Tutte's method,
their method may also produce drawings with exponential
area.
Indeed, we know of no such prior result, and, in fact,
Duncan {\it et al.}~\cite{duncan}
pose as an open problem whether there exists an algorithm
that produces straight-line drawings with vertices on the boundary
of a given convex polygon of polynomial area.

One motivation for this problem of prescribing the outer face of a planar drawing
comes from a common way of producing hand-drawn
planar representations of genus-$g$ graphs.
Namely, if we are given
a graph, $G$, embedded into a
genus-$g$ topological surface,
the surface may be cut along
the edges and vertices of $2g$ fundamental cycles in $G$ to form a topological disk (known as a \emph{canonical polygonal schema}), with a boundary that is made up
of $4g$ paths (with multiple copies of the vertices on the fundamental cycles).
Moreover, as shown by
Duncan {\it et al.}~\cite{duncan},
$G$ can be cut in this way so that
each of these $4g$ paths is chord-free, that is,
so that there are no edges
between two vertices strictly internal to the same path (other than
path edges themselves).

The standard way of drawing this unfolded version of
such an embedding, in the topology
literature (e.g., see~\cite{lpvv-ccpso-01}),
is to draw the disk as a convex polygon with each of its $4g$ boundary paths drawn as a straight line
segment: the geometric shape %of $f$
is used to make clear the pattern in which the surface was cut to form a disk.  Fortunately, given Tutte's seminal result, it is
possible to draw any chord-free canonical polygonal schema along the
boundary of a given convex polygon with $4g$ edges.
The drawback of using Tutte's algorithm for this purpose %, of course,
is that the resulting drawing may have exponential area.
Thus, we are interested in drawing the unfolded embedding in
polynomial area and in polynomial time.

\subsection{Our Results}
In this paper, we describe an algorithm for drawing a planar graph with a prescribed outer face shape.
The input consists of an embedded planar graph $G$, a partition of the outer face of the embedding into a set $\cal S$ of $k$ chord-free paths, and a $k$-sided polygon $P$; the output of our algorithm is a drawing of $G$ within $P$ with each path in $\cal S$ drawn along an edge of
$P$.
Given the above-mentioned prior result of Duncan {\it et al.}~\cite{duncan},
for finding chord-free canonical polygonal schemas,
our result implies that we can solve their open problem: any
graph $G$ combinatorially embedded
in a genus-$g$ surface has a polynomial-area straight-line planar drawing
of a canonical polygonal schema $S$ for $G$,
drawn as a $4g$-sided convex polygon $P$ with
the vertices of each path in $S$ drawn along
an edge of $P$.

\section {Preliminaries} \label {sec:prelims}

  In this paper, we show how to draw a graph with a given boundary with
  coordinates of polynomial magnitude.
  Before treating the main construction,
  though, we show in this section that
  we can equivalently state the problem
  in terms of the \emph{resolution} of the
  graph. Furthermore, we recall some known results and concepts.

  \subsection {Resolution}

    Instead of drawing a graph with integer coordinates of small total size, we
    will make a drawing with real coordinates that stays within a fixed region
    (inside the input polygon) with a large \emph {resolution}.

    Let $G$ be a graph that is embedded in $\R^2$ with straight line segments as
    edges. We define the \emph {resolution} of $G$ to be the shortest distance
    between either two vertices of the graph, or between a vertex and a
    non-incident edge.
    The \emph {diameter} of $G$ is the largest distance between two vertices of the graph.  
    
%    Here, two polygons are said to be \emph{combinatorially equivalent} if the ordering of the vertices along them is identical. \erin{replace with snap rounded scaled polygon?}
    
    We begin by establishing a relation between resolution and size, which basically says that drawing a graph $G$ with small diameter and large resolution also results in another drawing with integer coordinates and small size.
    Generally it may not be possible to scale a given input polygon such that its coordinates become integers, so we need to do some rounding. We say that two drawings of $G$ are \emph {combinatorially equivalent} if their topology is the same, and any collinear adjacent edges in one drawing are also collinear in the other.    
    We say two drawings are \emph {$\varepsilon$-equivalent} if the distance between the locations of each vertex of $G$ in the two drawings is at most $\varepsilon$.

    \begin {lemma} \label {lem:ressize}
      Let $G$ be a graph, and let $\Gamma$ be a drawing of $G$ without crossings, with constant resolution, and with diameter $D$. 
      Then there exists another drawing $\Gamma'$ of $G$ with integer coordinates and diameter $O(D^2)$, such that a scaled copy of $\Gamma'$ with diameter $D$ is both combinatorially equivalent and $O(1)$-equivalent to $\Gamma$.
    \end {lemma}

    \begin {proof}
      We may assume that the resolution of the initial drawing of $G$ is $1$:
      that is, no vertex is within unit distance of another vertex or edge.
      Scale $G$ by a factor of $3$,
      forming a drawing, $G'$, and consider the integer grid squares
      that contain each vertex of $G'$.
      If any set of vertices has its coordinates changed by at most one unit,
      this motion can only bring a vertex
      and an edge closer together by a distance of $2\sqrt 2$,
      less than the resolution (which is now $3$),
      so this motion cannot introduce crossings or change
      the combinatorial type of the drawing.
      At this stage, we move the vertices of the outer face to their
      nearest integer points (changing their coordinates by at most $1/2$),
      but we do not change the positions of the other vertices.

      Next, we scale the drawing again, by a factor of $\lceil 3D\rceil$. The vertices of the outer face remain on integer coordinates. The vertices in the interior of the drawing may be moved to any nearby integer point, without changing the combinatorial type of the drawing. It remains to choose integer coordinates for the vertices that lie on the sides of the outer face of $G$. For each such vertex $v$, on side $s$, the integer rounding of the endpoints of $s$ will have caused $s$ to move, but there will still exist a point $v'\in s$ whose coordinates (though not necessarily integers) are both within $3D/2$ of $v$. A $3D\times 3D$ box centered on $v'$ will consist entirely of points whose coordinates are within $3D$ of $v$ (equivalent to being within one unit of $v$ in $G'$), and (because of the way we rounded the endpoints of $s$ prior to the second scaling step) is guaranteed to contain an integer point $v''\in s$. By rounding $v$ to $v''$, and simultaneously rounding in the same way all the other points of the drawing that belong to the sides of the outer face, we obtain a drawing $G''$ that is combinatorially equivalent to $G$, with a combinatorially equivalent outer face, in which all coordinates are integers.
    \end {proof}

    Note that, for a fixed input polygon with non-integer vertex coordinates, this
    perturbation may slightly modify its shape, since it may not be possible to
    find a similar copy of the polygon with integer vertex coordinates.

    Now, let $Q$ be a set of points in the plane. We define the \emph {potential
    resolution} of $Q$ to be the resolution of the complete graph on $Q$.
    Similarly, for a polygon $P$, we define its potential resolution to be the
    potential resolution of its set of vertices.
    Clearly, the resolution of any drawing we can achieve will depend on the
    potential resolution of the input polygon, because the drawing could be
    forced to include any edge of the complete   graph.

    Next, we make an observation about the nature of the potential resolution of
    convex polygons.

    \begin {observation} \label {obs:convex-easier}
      If $P$ is a convex polygon, then the potential resolution of $P$ is the
      minimum over the vertices of $P$ of the distance between that vertex and the
      line through its two neighboring vertices.
    \end {observation}

% I think this is simple and obvious enough that we don't need a proof -- DE
%    \begin {proof}
 %   \end {proof}

    Thus, for a convex polygon $P$ to allow for a drawing of polynomial
    area in its interior, we insist that
    $P$ has a polynomially-bounded aspect ratio. It cannot be
    arbitrarily thin and still support a polynomial-area drawing in its
    interior.

  \subsection {Alpha Cuts}

    We now describe a useful property of the potential resolution of a convex
    polygon, namely that it can be ``distributed'' any way we want when cutting 
    the polygon into smaller parts.
    This will be made more precise later.
    We first make another observation about convex polygons.
    
    \begin {lemma} \label {lem:move-vertex}
      Let $P$ be a convex polygon, $v$ a vertex of $P$, $e$ an edge incident to $v
      $, and $\alpha \in (0, 1)$ a number. Let $P'$ be a copy of $P$ where $v$
      has been replaced by $v'$ by moving $v$ along $e$ over a fraction $\alpha$
      of the length of $e$. Then the potential resolution of $P'$ is at least $1 -
      \alpha$ times the potential resolution of $P$.
    \end {lemma}

    \begin {proof}
      Let $a, b, c$ be three vertices of $P$ in counterclockwise order, and
      consider the distance from $b$ to the line $\ell$ through $a$ and $c$, as in
      Figure~\ref {fig:ro-original}. By Observation~\ref {obs:convex-easier}, the
      potential resolution of $P$ is the minimum of this distance over all such
      triples of vertices.
      Now consider the situation in $P'$. If none of $a, b, c$ are equal to $v$,
      then clearly the distance in $P'$ is the same as in $P$.

      \vijfplaatjes {ro-original} {ro-case-br} {ro-case-al} {ro-case-ar}
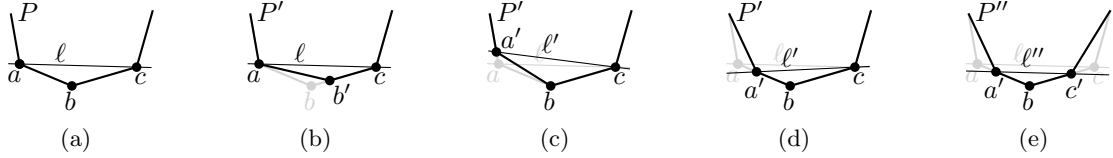
      {ro-case-arcl} {Cases to consider with respect to resolution with
      respect to a convex polygon.}

      If $v = b$, then the line $\ell$ through $a$ and $c$ is still the same as in
      $P$, see Figure~\ref {fig:ro-case-br}. The edge $e$ must be either between $
      a$ and $v$ or between $v$ and $c$; in both cases, moving $v$ along $e$
      changes the distance to $\ell$ linearly. Therefore, the distance from $b'$
      to $\ell$ is exactly $(1 - \alpha)$ times the distance from $b$ to $\ell$,
      which is bounded by $(1 - \alpha)$ times the potential resolution of $P$.

      If $v = a$, then there are two subcases depending on whether $e$ is between
      $v$ and $b$ or between $v$ and its predecessor.
      If $e$ goes between $v$ and the predecessor of $v$, as in Figure~\ref
      {fig:ro-case-al}, then the line $\ell'$ between $v'$ and $c$ rotates around
      $c$ as $v'$ moves over $e$.
      Because $P$ is convex, the distance from $b$ to $\ell'$ as a function of the
      position of $v'$ has no local minimum, so when $v'$ is at one of the
      endpoints of $e$ the distance is smaller than at any interior point.
      Therefore, the distance from $b$ to $\ell'$ is bounded by the potential
      resolution of $P$.
      If $e$ goes between $v$ and $b$ itself, as in Figure~\ref {fig:ro-case-ar},
      then let $c'$ be the point on a fraction $\alpha$ along the edge between $c$
      and $b$. Then distance from $b$ to the line $\ell'$ through $a'$ and $c$ is
      clearly larger than the distance from $b$ to the line $\ell''$ through $a'$
      and $c'$, which is exactly $(1 - \alpha)$ times the distance from $b$ to the
      line $\ell$ through $a$ and $c$, see Figure~\ref {fig:ro-case-arcl}. So in
      this case the distance is again bounded by $(1 - \alpha)$ times the
      potential resolution of $P$.

      The case where $v = c$ is symmetric to the case where $v = a$.
    \end {proof}
    
    Let $P$ be a convex polygon. We will show that we can cut $P$ into two smaller
    polygons, ``distributing'' its potential resolution in any way we want.

    We define an $\alpha$-cut of $P$ to be a directed line $\ell$ that splits $P$
    into two smaller polygons, such that if an edge $e$ of $P$ is intersected by $
    \ell$, the length of the piece of $e$ to the left of $\ell$ is $\alpha$ times
    the length of $e$, and the piece of $e$ to the right of $\ell$ is $(1-\alpha)$
    times the length of $e$.
    For a given convex polygon and two features of its boundary (either vertices
    or edges), there is a unique $\alpha$-cut that cuts the polygon through those
    two features in order.

    \begin {lemma} \label {lem:alphacut}
      Suppose we are
      given a convex polygon $P$ of resolution $d$, two features (either vertices
      or edges) of $P$, and a fraction $0 < \alpha < 1$.
      Let $\ell$ be the $\alpha$-cut through the two given features that cuts $P$
      into a piece $P_l$ to the left of $\ell$ and a piece $P_r$ to the right of $
      \ell$.
      Then the potential resolution of $P_l$ is at least $\alpha d$ unless the two
      features are two adjacent edges that meet to the right of $\ell$.
      Similarly, the potential resolution of $P_r$ is at least $(1-\alpha)d$
      unless the two features are two adjacent edges that meet to the left of $
      \ell$.
    \end {lemma}

    \begin {proof}
      We will argue about the potential resolution of $P_l$; the argument for $P_r
      $ is symmetric.
      We prove this lemma by applying Lemma~\ref {lem:move-vertex} to the
      new vertices of $P_l$. If both features where $\ell$ cuts through $P$ are
      vertices, then all vertices of $P_l$ are also vertices of $P$ and clearly
      the potential resolution can only become better. However, if one or both of
      the features are edges, then $P_l$ has one or two new vertices that are not
      part of $P$.
      Figure~\ref {fig:alphacut-1+alphacut-2+alphacut-2-adjacent} shows three different cases that can occur.
      To solve this problem, we first alter $P$ to a different
      polygon $P'$ that has the new vertices. Let $u'$ be the place where $\ell$
      enters $P$ and $u$ the closest vertex below $\ell$ along the boundary to it
      (possibly $u = u'$), and similarly let $v'$ be the place where $\ell$ exits
      $P$ and $v$ the closest vertex below $\ell$. Now, we create $P'$ by moving
      $u$ to $u'$ and $v$ to $v'$. Clearly, both will move a fraction $1 - \alpha$
      along their edges, so by Lemma~\ref {lem:move-vertex} $P'$ has a potential
      resolution of at most $\alpha$ times the potential resolution of $P$.
      Therefore, the potential resolution of $P_l$ can only be larger.

      The only exception is when $u = v$; in this case we cannot move the vertex
      to two new places simultaneously, but we have to create two new vertices.
      Indeed, the result is not true in that case, since the two new vertices
      can be arbitrarily close to each other as $\alpha$ comes arbitrarily
      close to $1$, so the resolution of $P'$ cannot be expressed in terms of
      $\alpha$, as can be seen in Figure~\ref {fig:alphacut-2-adjacent}.
    \end {proof}

    \drieplaatjes {alphacut-1} {alphacut-2} {alphacut-2-adjacent}
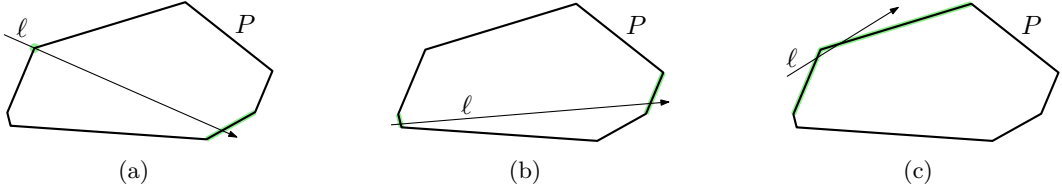
    {(a) An $\alpha$-cut through a vertex and an edge. (b) An $\alpha$-cut through two non-adjacent
    edges. (c) An $\alpha$-cut through two adjacent edges.}

  \subsection {Combinatorial Embeddings}

    Let $G=(V,E)$ be a plane graph.
    That is, we consider the combinatorial structure of $G$'s
    embedding to be fixed,
    but we are free to move its vertices and edges around.
    Let $F$ be the set of faces of $G$, excluding the outer face.
%    Apart from its set of vertices, $V$, and
%    set of edges, $E$, $G$ also has a set of faces, $F$
%    (excluding the outer face), in this case.
    We make some definitions about faces.
    We say that a subset $F' \subset F$ induces a subgraph $\fsg {F'}$ of $G$
    that consists of all vertices and edges that are incident to the faces
    in $F'$.
    A subset $F' \subset F$ is said to be \emph {vertex-connected} if $\fsg {F'}$
    is connected; it is said to be \emph {edge-connected} if the dual graph
    induced by the dual vertices of $F'$ is connected. In other words, faces
    that share an edge are both edge-connected and vertex-connected, but faces
    that share only a vertex are only vertex-connected.

    We recall a lemma from~\cite{duncan},
    rephrased in terms of the faces of the
    graph:

    \begin{lemma}
    \label{lem:river}
      Given an embedded plane graph $G$ that is fully triangulated except for the
      external face and two edges $e_1$ and $e_2$ on that external face,
      it is possible to partition the faces of $G$ into three sets
      $F_1, F_2, R \subset F$
      such that:
      \begin{enumerate}
        \item All vertices of $G$ are in either \fsg {F_1} or \fsg {F_2}.
        \item $R$ is edge-connected and contains the faces incident to $e_1$
          and $e_2$.
        \item $F_1$ and $F_2$ are both vertex-connected.
        \item The edge-connected components of $F_1$ and $F_2$ all share an edge
          with the outer face of $G$.
      \end{enumerate}
    \end{lemma}
%    \maarten {I mutilated this lemma quite a bit; should we actually
%    prove that this is equivalent to the old representation?}

    Intuitively, $R$ is a path of faces that goes from $e_1$ to $e_2$ and that
    splits the remaining faces into two sets $F_1$ and $F_2$.

\section {Drawing a Graph with a Given Boundary}

  We are now ready to formally state the problem and describe the algorithm to
  solve it.

  \subsection {The Problem}

    Let $G$ be a triangulated planar graph with a given combinatorial embedding,
    and let $B$ be the cycle that bounds the outer face of $G$.
    Let $f$ be a map from a subset of the vertices of $B$ to points in the
    plane, such that these points are in convex position and their order along
    their convex hull is the same as their order along $B$.

    We say that a map $g$ from all vertices of $G$ to points in the plane
    \emph {respects} $f$ when:
    \begin {enumerate}
      \item The vertices mapped by $f$ are also mapped by $g$ to the same points;
      these define a convex polygon $P$.
      \item The remaining vertices of $B$ are mapped to the corresponding edges of $P$.
      \item The remaining vertices of $G$ are mapped to the interior of $P$.
      \item If all edges are drawn as straight line segments, they cause no
      crossings or incidences not present in $G$.
    \end {enumerate}

    An example of a respectful embedding was shown in Figure~\ref
    {fig:intro-input-G+intro-input-P+intro-output}.

    The input to our problem is a pair $(G, f)$.
    We will use the notations $B$ and $P$ as above. We will further define $F$
    to be the set of faces of $G$, excluding the outer face.
    We define $s = |F|$, the number of internal faces,
    to be the \emph {size} of the problem.
    We define $d$ to be the \emph {resolution} of the problem, which is the
    potential resolution of $P$ (recall, that is the resolution of the complete
    graph on the corners of $P$).
    Our goal is to compute a mapping $g$ that respects $f$ and such that the
    resolution of the embedded graph is bounded by some function of $s$ and $d$.

    Observe that it will not be possible to do so when there are any edges in $G$
    between two vertices that have to be on the same edge of $P$. Therefore, we
    call a problem \emph {invalid} if this is the case.
    %We also say a problem is invalid if $G$ is not triconnected.
    We will show that for any
    valid problem, we can find an embedding with a polynomially bounded resolution.
    %of ... in ... time.

  \subsection {The Main Idea}

    We want to solve the problem using divide-and-conquer.
    The idea is to divide $P$ into smaller convex polygons, and $F$ into smaller
    sets of faces, and map each subset of faces to one of the smaller regions.
    Then we need to decide which vertices of $G$ are mapped to the new corners
    of the smaller regions, and solve the subproblems.

    A first idea would be
    to find a path in $G$ between two vertices of $B$, and lay that out on a
    straight line, resulting in a split of $P$ into two smaller polygons and solve
    the two subproblems. There are two issues with this approach though. First, the
    vertices on the new straight line have to be placed the same way in the two
    subproblems, which means they are not independent. Second, if there are any
    chords on this path one of the subproblems will become invalid.

    To avoid these issues, we will not split along a single path, but along two
    paths next to each other. The region between these two paths, which we call a
    \emph {river}, has a controlled structure, which means that we can always
    complete the interior independently of how the vertices on the edges were
    placed. Furthermore, if these paths have any chords, we shortcut them along
    the chords and show how to deal with the added complexity of the river.
    Because the river may touch the boundary of $P$ in more places, the problem
    may be decomposed into more than two subproblems.
    Figure~\ref {fig:ex-river} shows an example instance,
    and Figure~\ref {fig:ex-reduced} shows
    a possible decomposition where some vertices on the boundary of $P$ have been
    fixed, and the paths between them are made straight.

    \drieplaatjes {ex-river} {ex-chords} {ex-reduced}
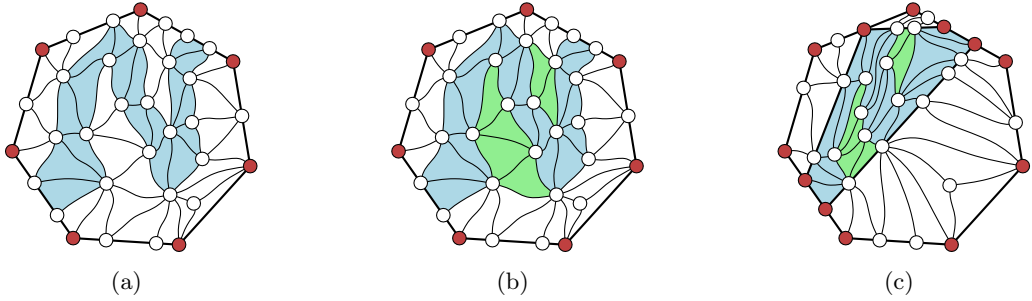
    {(a) A ``river'' (a path in the dual graph that does not reuse any vertices of
    the primal graph) between two edges on $B$. (b) The river banks have chords, and so 
    we include the area behind the chords in the river. (c) We fix the vertices on
    the river boundary that are on $B$, and draw the rest of the river boundary
    straight. This results in three smaller problems, plus the area of the river
    itself.}

    We assume the input is a valid problem with size $s$ and resolution $d$.
    We will keep as an invariant the ratio $d/s$, and show in
    \ifFull Section~\ref {sec:split} \else the next paragraph \fi
    how to
    subdivide a problem into smaller valid problems with the same (or better)
    ratio, plus an extra region (the river). We then recursively solve the
    independent subproblems, which results in a placement of all vertices
    that are not in the interior of the river. Finally, we show in
    \ifFull Section~\ref {sec:river} \else the paragraph after that \fi
    how to we place the vertices inside the river.

  \subsection {Splitting a Problem} \label {sec:split}

    Let $(G, P)$ be a valid problem of size $s$ and resolution $d$,
    and suppose that $P$ has at least four sides.

    Let $e_1$ and $e_2$ be two edges of $B$ that lie on two sides of $P$ that
    are not consecutive. Note that the endpoints of $e_1$ and $e_2$ are not
    necessarily fixed yet.
    Now, by Lemma~\ref {lem:river}, there exists a path of faces
    that connects $e_1$ to $e_2$, such that the boundary of this path does not
    have any repeated vertices. Let $R$ be the union of the faces of $F$ on
    this path. We call $R$ a \emph {river}; Figure~\ref {fig:ex-river} shows
    an example.
    This river may touch $B$ in other points than $e_1$ and $e_2$, so it can
    subdivide the faces of $F$ into any number of edge-connected subsets
    (apart from the river itself). We will assign a separate subproblem to
    each such edge-connected component.

    We would like to straighten the banks of the river, but this may lead to
    invalid subproblems if these banks have any chords. Therefore, we identify
    any chords that the river has
    (note that they can only appear on the outside of the river, since the river forms a dual path),
    and we add the faces of $F$ behind those chords to $R$.
    Similarly,
    if one of the paths touches a side of $P$ more than once, it would create a
    subproblem that would be flattened. To avoid that, we also incorporate such a
    region into the river (even though the straight side that lies alongside $P$
    is not necessarily a chord).

    Next, we count the numbers of faces in the river, as well as those
    in the parts outside the river. Then we fix the vertices where the river
    touches $P$ by cutting off the subproblems, using $\alpha$-cuts where $\alpha$
    is the fraction of faces inside the subproblem.
    Now, by Lemma~\ref {lem:alphacut},
    if we have a problem with parameters $s$ and $d$, we will
    construct subproblems with the same (or better) ratio $d/s$.
    Finally we straighten the new
    banks of the river, so that the subproblems have proper
    convex boundaries. Figure~\ref {fig:ex-river+ex-chords+ex-reduced} shows
    an example.

    \begin {lemma} \label {lem:divide}
      Given a valid problem $(F, P)$ where $P$ has at least four sides,
      We can subdivide $F$ and $P$ into disjoint sequences $F_1, F_2, \ldots, F_h$
      and $P_1, P_2, \ldots, P_h$ such that each $(\fsg {F_i}, P_i)$ is a valid
      subproblem with ratio $d/s$, and such that the remainders
      $F' = F \setminus \bigcup F_i$ and $P' = P \setminus \bigcup P_i$ have
      the following properties:
      \begin {enumerate}
        \item $F'$ and $P'$ also have ratio $d/s$.
        \item The vertices of $\fsg {F'}$ that are not vertices of
        $\fsg {\bigcup F_i}$ form internally 3-connected components that
        share at least two vertices of        $\fsg {\bigcup F_i}$ .
      \end {enumerate}
    \end {lemma}

    \begin {proof}
      For the first part of the lemma, we need to show that the subproblems
      $(F_i, P_i)$ are valid and have a resolution/size ratio at least as
      good as $d/s$.
      First, we define the polygons $P_i$ by applying Lemma~\ref {lem:alphacut}
      to $P$ with $\alpha = s_i/s$ (that is, the fraction of
      faces in $F$ that is in $F_i$). The lemma ensures that the new polygons
      have potential resolution at least $d_i \geq \alpha d = s_i d/s$, so clearly
      $d_i / s_i \geq d / s$ as required.
      Second, recall that a subproblem is valid if it does not have any chords
      between two vertices that have to be drawn on the same side of $P_i$.
      For those sides of $P_i$ that are part of $P$, we already know there are
      no chords because $(G, P)$ was valid. For the new sides, we explicitly
      added all faces behind chords to the river $R = F'$.

      For the second part of the lemma, we need to show that $R$
      has the right ratio and that the internal vertices of the river form
      3-connected components that share two or more vertices with the boundary
      of the river.
      If we denote $s_R = |R|$ to be the size of the river and $d_R$ to be the
      potential resolution of the region $P'$ in which it is to be drawn, then
      by Lemma~\ref {lem:alphacut}, the potential resolution of $R$ after
      repeatedly slicing off subpolygons is at least $d_R = d\Pi (1-s_i/s)
      \geq d(1-\Sigma(s_i/s)) = s_R d / s$, so $d_R / s_R \geq d / s$.
      Finally, since we chose $R$ according to Lemma~\ref {lem:river}, all
      edge-connected pieces of $F$ outside $R$ share an edge with the outer
      face. In particular, this means that the boundary of $R$ does not touch
      itself and that any subgraphs sliced off by chords are 3-connected and
      share exactly two vertices with the boundary of $R$.
      Furthermore, if the boundary of $R$ touches the same side of $P$ multiple
      times, then the edge-connected components of $F$ between that side of $P$
      and $R$ are also 3-connected and share at least two vertices with the new
      boundary of $R$.
    \end {proof}

    When $P$ has only 3 sides, we cannot choose two edges $e_1$ and $e_2$ on
    non-adjacent sides of $P$. However, we can still use the same basic idea; 
    we just have to be careful because of the special case in Lemma~\ref
    {lem:alphacut}. So, let $c$ be a corner of $P$ and let $e_1$ and $e_2$
    be edges of $B$ on the sides incident to $c$. The lemma does not give a
    bound on the resolution of the region on the far side of $c$. So, let
    $e_1$ and $e_2$ be the edges furthest away from $c$. Since $P$ is a
    triangle, the two vertices of $B$ on the far side on $e_1$ and $e_2$
    are the other two corners of $P$, and they are joined by a side of $P$.
    This means that the region between the river and this side will be
    included into the river, and there will be no subproblem on the far
    side of $c$.

  \subsection {Fixing a River} \label {sec:river}

    \drieplaatjes {ff-input} {ff-bumps} {ff-flat}
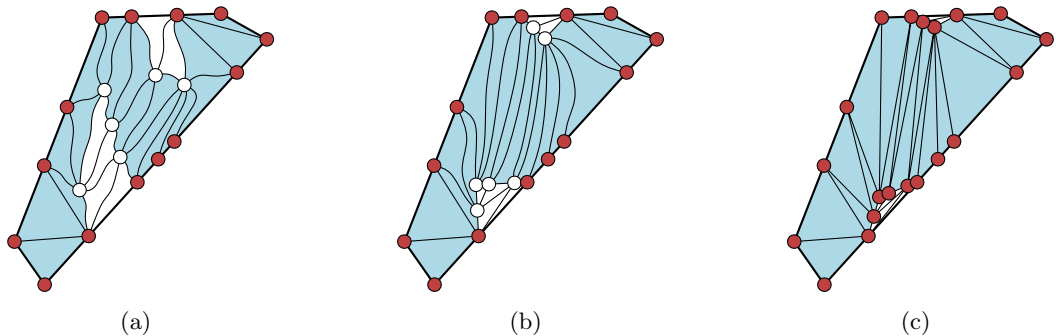
    {(a) The interior of a river, after all vertices on its boundary have already
    been fixed by recursive calls to the split algorithm. (b) Because of
    the structure given by the river, we can identify small areas inside the
    river that we draw using the de Fraysseix-Pach-Pollack algorithm. (c) If we
    flatten the triangular drawings enough, every vertex is able to see every
    other one (in fact we need slightly more, namely that no visibility ray comes
    too close to any vertex). Note that in the figure the drawing is not flat
    enough for that, but otherwise the structure would become too hard to see.
    In fact, in this particular case no flattening at all would be required.}

    It remains to show how to place the interior vertices of a river after all vertices on its boundary have been placed recursively.
    Again, we are given a graph $G$ and a polygon $P$ that it has to be drawn in
    ($P$ is the boundary of the river, and $G$ is the part of the graph
    that has to be drawn in it),
    but there are two important differences with the initial problem:
    First, now we know that all vertices of $B$ have already been fixed
    (not only those on the corners but also those on the boundary of $P$).
    Second, we know that the remaining vertices of $G$ have a very specific
    structure, namely, they  form internally 2-connected components that
    share at least two vertices with the vertices fixed along an edge.
    Furthermore, since all fixed vertices have been placed using the algorithm
    above, they will never be closer than $d/n$ to each other.
    This means we can draw these components using the algorithm by de Fraysseix,
    Pach and Pollack~\cite{dpp}, and rotate and scale them to fit inside $P$. We
    can then flatten them more such that all remaining edges (between fixed
    vertices on the boundary of $P$ or vertices on the de Fraysseix-Pach-Pollack
    drawings) can be drawn with straight line segments.
    Figure~\ref {fig:ff-input+ff-bumps+ff-flat} shows an example.

    \begin {lemma} \label {lem:internal}
      Given a river placement problem, of size $s$ and resolution $d$, we can lay
      out the graph with a resolution of $\Omega(d/n^3)$.
    \end {lemma}

    \begin {proof}
      The polygon $P$ that forms a river and the graph $G$ to be drawn in it
      are formed by subdividing a bigger problem according to Lemma~\ref
      {lem:divide}.
      Then, the vertices on the boundary of $P$ are placed during recursive calls
      to smaller problems. All vertices have to be corners of the polygons of at
      least one such subproblem, and by our invariant all these problems have
      ratio $d/s$, so the distance between any two fixed vertices cannot be
      smaller than $d/n$.
      Furthermore, Lemma~\ref {lem:divide} tells us that
      the vertices of $G$ can be grouped into
      a number of subsets $V_1, V_2, \ldots, V_h$ such that each $V_i$ is
      internally biconnected, and there is a
      sequence of at least two vertices in $V_i$ that is in $B$, so that have
      been fixed  on the boundary of $P$.
      Note that there will be exactly two if such a component came from a chord on
      the river bank, but there can be more if it came from the river touching a
      side of $P$ multiple times.
      Then we can place these subgraphs using the de Fraysseix-Pach-Pollack
      algorithm, starting from the vertices that are already fixed (at distance
      at least $d/n$) and adding the remaining $O (|V_i|)$ vertices one by one
      using $45^\circ$ edges. This results in a drawing with resolution
      which can be roughly bounded by $d/n^2$.
%      \maarten {Should we apply Lemma 1 in the opposite direction here,
%      or is it clear enough?}
      Then, it is sufficient to squeeze them by a factor of $n$ to make sure
      that they do not block any potential edges, and a further factor 2
      to make sure that the tips of the de Fraysseix-Pach-Pollack drawings are in
      fact far enough away from these potential edges, guaranteeing a good
      resolution.
      This means the final resolution of the drawing is $\Omega(d/n^3)$.
    \end {proof}

  \subsection {Putting it Together}

    To conclude, Lemmas~\ref {lem:divide} and~\ref {lem:internal} together imply:

    \begin {theorem} \label {thm:fixeddrawing}
      Given a plane graph $G$ with $n$ vertices, a convex polygon $P$ with $k$ corners
      and potential resolution $d$,  and a map $f$ that maps $k$
      vertices on the outer face of $G$ to the $k$ corners of $P$,
      we can draw a $G$ in $P$ respecting $f$ using resolution $\Omega(d/n^3)$.
    \end {theorem}

    Note that by Lemma~\ref {lem:ressize}, we can rephrase this in terms of the
    more standard \emph {area} of a drawing when all coordinates are integer.

    \begin {corollary}
      Given a plane graph $G$ with $n$ vertices, a convex polygon $P$ with $k$ corners,
      at integer coordinates and diameter $D$, and a map $f$ that maps $k$
      vertices on the outer face of $G$ to the $k$ corners of $P$,
      we can draw the graph $G$ in a scaled copy $P'$ of $P$ that has diameter $O (D^4n^6)$, such that the drawing 
      respects $f$ and uses only integer coordinates for the vertices of $G$.
%      Given a plane graph $G$ with $n$ vertices, a polygon $P$ with $k$ corners
%      at integer coordinates, and a map $f$ that maps $k$ vertices of the outer
%      face of $G$ to the corners of $P$,
%      we can draw $G$ in $P$ respecting $f$ and using only integer coordinates
%      for the vertices of $G$
%      if the potential resolution of $P$ is larger than $cn^6$ for some constant
%      $c$.
%      \maarten {Because if $d = n^3$ then by the theorem above we can draw it with constant resolution, but then Lemma 1 squares all distances, so also the resolution of the drawing. Does this make sense?}
%\david{There's a problem here. The assumption that $P$ has corners at integer coordinates isn't good enough: we need the slopes of the edges to be rational numbers with small enough numerators and denominators, so that we can be sure that there are ``enough'' integer points along the sides of $P$. Or, as in Lemma~\ref{lem:ressize}, we need to start with a small integer polygon and then rescale it to be a larger one with lots of integer points along its sides.}
    \end {corollary}

    \begin {proof}
      First of all, if $P$ has only integer coordinate vertices and diameter $D$,
      then its potential resolution is at least $1/D$. To see this, consider
      a triangle formed by any three vertices of $P$: this triangle has area at
      least $1/2$, and in any direction its base is at most $D$ so its height
      must be at least $1/D$.

      Now, by Theorem~\ref {thm:fixeddrawing}, we can draw $G$ inside $P$ with
      resolution $\Omega(d/n^3) = \Omega (1/Dn^3)$. Then we can blow up the drawing by
      a factor $Dn^3$, which results in a polygon of diameter $D^2n^3$
      and at least constant resolution. By Lemma~\ref {lem:ressize}, there now also
      exists a drawing of $G$ in a polygon $P'$ of diameter $O (D^4n^6)$
      in which all vertices are drawn with integer coordinates.
    \end {proof}

\section {Application to Drawing Graphs of Genus $g$}

  As mentioned in the introduction,
  graphs of genus $g$ are often drawn in the plane by drawing their \emph {polygonal schema} in a
  prescribed convex polygon. Using a canonical polygon schema
  allows us to draw this outer face
  as a
  regular $4g$-gon that has some pairs of edges identified,
  and vertices on those edges
  duplicated.
  Given previous work by Duncan {\it et al.}~\cite{duncan},
  which gives us a chord-free polygonal schema derived from a graph
  $G$ combinatorially-embedded in a genus-$g$ surface, we can
  complete a straight-line drawing of $G$ using a given regular $4g$-gon as
  its external face, by
  applying Theorem~\ref {thm:fixeddrawing} and rounding the
  coordinates.
  Since a regular $4g$-gon with diameter $1$ has potential resolution $\Theta(1/g^2)$, this results in a
  drawing with resolution $\Omega(1/g^2n^3)$.

Of course, there is a
  slight issue with using a regular $4g$-gon:
  not every regular $k$-gon can be embedded with fractional
  coordinates.
  So, such a drawing will not fit exactly on an integer grid no matter how big the
  integers can be.
  Thus, we either have to allow for non-integer coordinates or allow
  for a slight (possibly imperceptible) perturbation of the vertex
  coordinates.

\section{Conclusion and Open Problems}
We have given an algorithm to draw any combinatorially-embedded
planar graph with a prescribed convex shape as its outer face and
polynomial area, with respect to the potential resolution of that
shape. That is, if the given convex shape has a polynomially-bounded
aspect ratio, then we can draw the graph $G$ in its interior using
polynomial area.
We have not made a strenuous attempt to optimize the exponent in this
area bound. So a natural open problem is to determine the upper and
lower bound limits of this function.

In addition, with respect to drawings of genus-$g$ graphs using a
canonical polygonal schema, although our construction guarantees that
copies of corresponding
vertices appearing in multiple boundary paths
will be drawn in the same relative order,
it does not guarantee that they will be drawn with the same
inter-path distances.
So another open problem is whether one can extend our algorithm to
draw such paths with matching inter-path distances for corresponding
vertices.

\section*{Acknowledgments}
This research was supported in part by the National Science
Foundation under grant 0830403, and by the
Office of Naval Research under MURI grant N00014-08-1-1015.

\raggedright
\bibliographystyle{abuser}
\bibliography{drawingggg}

\begin{thebibliography}{10}

\bibitem{br-scdpg-06}
I.~B{\'a}r{\'a}ny and G.~Rote.
\newblock {Strictly convex drawings of planar graphs}.
\newblock {\em Documenta Mathematica} 11:369{--}391, 2006,
  \href{http://arxiv.org/abs/cs/0507030}{arXiv:cs/0507030},
  \url{http://www.math.uiuc.edu/documenta/vol-11/13.html}.

\bibitem{bh-olpgf-87}
B.~Becker and G.~Hotz.
\newblock {On the optimal layout of planar graphs with fixed boundary}.
\newblock {\em SIAM J. Comput.} 16(5):946{--}972, 1987,
  \href{http://dx.doi.org/10.1137/0216061}%
{doi:10.1137/0216061}.

\bibitem{cgt-cdgtt-96}
M.~Chrobak, M.~T. Goodrich, and R.~Tamassia.
\newblock {Convex drawings of graphs in two and three dimensions}.
\newblock {\em Proc. 12th ACM Symp. Comput. Geom.}, pp.~319{--}328, 1996,
  \href{http://dx.doi.org/10.1145/237218.237401}%
{doi:10.1145/237218.237401}.

\bibitem{ck-cgd3c-97}
M.~Chrobak and G.~Kant.
\newblock {Convex grid drawings of 3-connected planar graphs}.
\newblock {\em Internat. J. Comput. Geom. Appl.} 7(3):211{--}223, 1997,
  \href{http://dx.doi.org/10.1142/S0218195997000144}%
{doi:10.1142/S0218195997000144}.

\bibitem{cp-ltadp-95}
M.~Chrobak and T.~H. Payne.
\newblock {A linear-time algorithm for drawing a planar graph on a grid}.
\newblock {\em Inform. Process. Lett.} 54(4):241{--}246, 1995,
  \href{http://dx.doi.org/10.1016/0020-0190(95)00020-D}%
{doi:10.1016/0020-0190(95)00020-D}.

\bibitem{dh-dgnus-96}
R.~Davidson and D.~Harel.
\newblock {Drawing graphs nicely using simulated annealing}.
\newblock {\em ACM Trans. Graph.} 15(4):301{--}331, 1996,
  \href{http://dx.doi.org/10.1145/234535.234538}%
{doi:10.1145/234535.234538}.

\bibitem{d-gdt-10}
R.~Dhandapani.
\newblock {Greedy drawings of triangulations}.
\newblock {\em Discrete Comput. Geom.} 43(2):375{--}392, 2010,
  \href{http://dx.doi.org/10.1007/s00454-009-9235-6}%
{doi:10.1007/s00454-009-9235-6}.

\bibitem{dett-gd-99}
G.~Di~Battista, P.~Eades, R.~Tamassia, and I.~G. Tollis.
\newblock {\em {Graph Drawing}}.
\newblock Prentice Hall, Upper Saddle River, NJ, 1999.

\bibitem{duncan}
C.~A. Duncan, M.~T. Goodrich, and S.~G. Kobourov.
\newblock {Planar drawings of higher-genus graphs}.
\newblock {\em Proc. 17th Int. Symp. Graph Drawing}, pp.~45{--}56.
  Springer-Verlag, LNCS 5849, 2010,
  \href{http://dx.doi.org/10.1007/978-3-642-11805-0\_7}%
{doi:10.1007/978-3-642-11805-0\_7},
  \href{http://arxiv.org/abs/0908.1608}{arXiv:0908.1608}.

\bibitem{f-slrpg-48}
I.~F{\'a}ry.
\newblock {On straight-line representation of planar graphs}.
\newblock {\em Acta Sci. Math. (Szeged)} 11:229{--}233, 1948.

\bibitem{dpp}
H.~de~Fraysseix, J.~Pach, and R.~Pollack.
\newblock {How to draw a planar graph on a grid}.
\newblock {\em Combinatorica} 10(1):41{--}51, 1990,
  \href{http://dx.doi.org/10.1007/BF02122694}%
{doi:10.1007/BF02122694}.

\bibitem{fr-gdfdp-91}
T.~M.~J. Fruchterman and E.~M. Reingold.
\newblock {Graph drawing by force-directed placement}.
\newblock {\em Softw. Pract. Exp.} 21(11):1129{--}1164, 1991,
  \href{http://dx.doi.org/10.1002/spe.4380211102}%
{doi:10.1002/spe.4380211102}.

\bibitem{ggk-mdafd-04}
P.~Gajer, M.~T. Goodrich, and S.~G. Kobourov.
\newblock {A multi-dimensional approach to force-directed layouts of large
  graphs}.
\newblock {\em Comput. Geom. Theory Appl.} 29(1):3{--}18, 2004,
  \href{http://dx.doi.org/10.1016/j.comgeo.2004.03.014}%
{doi:10.1016/j.comgeo.2004.03.014}.

\bibitem{k-dpguc-96}
G.~Kant.
\newblock {Drawing planar graphs using the canonical ordering}.
\newblock {\em Algorithmica} 16(1):4{--}32, 1996,
  \href{http://dx.doi.org/10.1007/BF02086606}%
{doi:10.1007/BF02086606}.

\bibitem{lpvv-ccpso-01}
F.~Lazarus, M.~Pocchiola, G.~Vegter, and A.~Verroust.
\newblock {Computing a canonical polygonal schema of an Orientable Triangulated
  Surface}.
\newblock {\em Proc. 17th ACM Symp. Comput. Geom.}, pp.~80{--}89, 2001,
  \href{http://dx.doi.org/10.1145/378583.378630}%
{doi:10.1145/378583.378630}.

\bibitem{s-epgg-90}
W.~Schnyder.
\newblock {Embedding planar graphs on the grid}.
\newblock {\em Proc. 1st ACM-SIAM Symp. Discrete Algorithms}, pp.~138{--}148,
  1990, \url{http://portal.acm.org/citation.cfm?id=320191}.

\bibitem{s-cm-51}
S.~K. Stein.
\newblock {Convex maps}.
\newblock {\em Proc. Amer. Math. Soc.} 2(3):464{--}466, 1951,
  \href{http://dx.doi.org/10.1090/S0002-9939-1951-0041425-5}%
{doi:10.1090/S0002-9939-1951-0041425-5}.

\bibitem{sm-gdmsm-95}
K.~Sugiyama and K.~Misue.
\newblock {Graph drawing by the magnetic spring model}.
\newblock {\em J. Visual Lang. Comput.} 6(3):217{--}231, 1995,
  \href{http://dx.doi.org/10.1006/jvlc.1995.1013}%
{doi:10.1006/jvlc.1995.1013}.

\bibitem{t-crg-60}
W.~T. Tutte.
\newblock {Convex representations of graphs}.
\newblock {\em Proc. London Math. Soc.} 10(38):304{--}320, 1960,
  \href{http://dx.doi.org/10.1112/plms/s3-10.1.304}%
{doi:10.1112/plms/s3-10.1.304}.

\bibitem{t-hdg-63}
W.~T. Tutte.
\newblock {How to draw a graph}.
\newblock {\em Proc. London Math. Soc.} 13(52):743{--}768, 1963,
  \href{http://dx.doi.org/10.1112/plms/s3-13.1.743}%
{doi:10.1112/plms/s3-13.1.743}.

\bibitem{w-bzv-36}
K.~Wagner.
\newblock {Bemerkungen zum Vierfarbenproblem}.
\newblock {\em Jber. Deutsch. Math.-Verein.} 46:26{--}32, 1936.

\end{thebibliography}

\end{document}